\documentclass[11pt,reqno]{amsart}
\usepackage{amssymb}

\usepackage{enumitem}
\usepackage{amsthm,amsfonts,amssymb,euscript,color,bbm}
\usepackage{comment}
\usepackage{mathtools}
\usepackage{stmaryrd}
\usepackage{mathrsfs}
\usepackage{pifont}
\usepackage{amsfonts,amsmath,amssymb,amsthm,amscd,cancel}
\usepackage{graphicx}
\usepackage{verbatim}
\usepackage{mathrsfs}
\usepackage{fancyhdr}
\usepackage{footnpag,footmisc}
\usepackage[normalem]{ulem}

\usepackage{cite}
\usepackage{hyperref}
\usepackage{engord}
\usepackage[displaymath]{lineno}

\theoremstyle{definition}
\newtheorem{theorem}{Theorem}[section]
\newtheorem{lemma}{Lemma}[section]

\newtheorem{remark}{Remark}[section]

\allowdisplaybreaks[4]

\DeclareMathAlphabet{\mathsfsl}{OT1}{cmss}{m}{sl}
\numberwithin{equation}{section}

\newcommand{\D}{\mathrm{d}}

\def\Lb{\underline{L}}

\def\omegab{\underline{\omega}}

\def\ub{\underline{u}}
\def\Cb{\underline{C}}

\newcommand{\Db}{\underline{D}}

\newcommand{\hb}{\underline{h}}

\def\nablas{\mbox{$\nabla \mkern -13mu /$ }}

\def\gs{\mbox{$g \mkern -9mu /$}}

\def\trs{\mbox{$\mathrm{tr} \mkern -9mu /$}}

\begin{document}

\title{Scalar perturbations to naked singularities of perfect fluid}

\author[Junbin Li]{Junbin Li}
\address{Department of Mathematics, Sun Yat-sen University, Guangzhou, China}
\email{lijunbin@mail.sysu.edu.cn}

\author[Xi-Ping Zhu]{Xi-Ping Zhu}
\address{Department of Mathematics, Sun Yat-sen University, Guangzhou, China}
\email{stszxp@mail.sysu.edu.cn}

 \begin{abstract}
  
 In this paper, we study the instability of naked singularities arising in the Einstein equations coupled with isothermal perfect fluid.  We show that the spherically symmetric self-similar naked singularities of this system, are unstable to trapped surface formation, under $C^{1,\alpha}$ perturbations of an external massless scalar field.  We viewed this as a toy model in studying the instability of these naked singularities under gravitational perturbations in the original Einstein--Euler system which is non-spherically symmetric. 
  
  \end{abstract}

\maketitle


\setcounter{tocdepth}{1}


\section{Introduction}

\subsection{Previous works and a rough version of the main result} 

The weak cosmic censorship conjecture, one of the main conjectures in general relativity, asserts that naked singularities will not appear in gravitational collapse generically.  The original statement proposed by Penrose \cite{Pen69}  said that naked singularities will not appear in any cases, but soon examples of naked singularities in gravitational collapses were found in spherically symmetric massless scalar field collapsing (for example \cite{Cho93} and \cite{Chr94}). To insist on the validity of weak cosmic censorship, one hopes to prove  that naked singularities are unstable under small perturbations, and therefore can only appear in very exceptional cases.

Soon after the discovery of naked singularities, Christodoulou \cite{Chr99} was able to show all possible naked singularities (not only those constructed in \cite{Chr94}) are unstable in the context of spherically symmetric solutions of Einstein--scalar field equations. Hence the weak cosmic censorship conjecture is true in this context, which is also studied in spherically symmetric collapse of charged scalar field  by An--Tan \cite{A-T24}. In \cite{Liu-Li18}, we were able to give a constructive argument based on apriori estimates of Christodoulou's proof in \cite{Chr99}, instead of contradiction argument there,  and show that  \cite{Li-Liu22} all possible spherically symmetric naked singularities are unstable to trapped surface formation under small $C^1$ gravitational perturbations (which must be outside spherically symmetric by Birkhoff Theorem).  A recent work by An \cite{An24} showed that when the  background naked singularity solutions are those constructed by Christodoulou in \cite{Chr94}, which are continuously self-similar, fully anisotropic apparent horizon (not only trapped surfaces) can form under small gravitational perturbations in scale--critical norm.

In this paper we focus on the naked singularity model which is not of zero rest mass. The main difference from the before mentioned systems is that the matter field itself has a slower  wave speed than light. Christodoulou \cite{Chr84} showed that spherically symmetric collapse of inhomogenuous dust cloud starting at rest would lead to stable formation of naked singularities, and more general  cases were considered (see for example \cite{S-J96} and the references therein). This model may not be physical since pressure will not vanish at the final stage of a collapsing star. When the pressure comes in, the corresponding Einstein--Euler system is rather complicated, a general picture of gravitational collapse cannot be easily obtained. Nevertheless, naked singularities in Einstein--Euler were still found by Ori--Piran (see \cite{O-P90} and the references therein) by numerical method, assuming spherical symmetry and continuous self-similarity of the spacetime, and the rigorous proof of this existence was provided recently by Guo--Hazic--Jang \cite{G-H-J23}. 

It is then natural to ask about the instability of such kind of naked singularities, arising in gravitational collapse of (spherically symmetric) perfect fluid. At first glance, this problem is about the long time dynamic of self-gravitating relativistic fluid, which is not  well-understood even in spherical symmetry in the literature. Nevertheless, we will try to study whether we can produce instability from the left hand side of the Einstein equations, i.e., the gravity part, as we are treating the whole Einstein--Euler system.  The price is that we should go beyond spherical symmetry due to Birkhoff Theorem. To illustrate our main ideas, we will study the instability under spherically symmetric perturbations from an external massless real scalar field in this paper. As stated in Christodoulou's survey \cite{Chr99cqg}, this is a good simplified model problem of Einstein's theory of gravitation in case we impose spherical symmetry. In a further work \cite{L-Z3}, we will drop the scalar field and study the original Einstein--Euler system.

More precisely, we will study the following system of Einstein equations
\begin{equation}\label{EinsteinSE}\mathbf{Ric}_{\alpha\beta}-\frac{1}{2}\mathbf{R}g_{\alpha\beta}=2\mathbf{T}_{\alpha\beta}=2\mathbf{T}_{\alpha\beta}^{sc}+2\mathbf{T}_{\alpha\beta}^{fl},\end{equation}
where the energy momentum tensors of the coupled scalar field and perfect fluid are
$$\mathbf{T}_{\alpha\beta}^{sc}=\partial_\alpha\phi\partial_\beta\phi-\frac{1}{2}g_{\alpha\beta}g^{\mu\nu}\partial_\mu\phi\partial_\nu\phi$$
and
$$\mathbf{T}_{\alpha\beta}^{fl}=(\varrho+p) U_\alpha U_\beta+pg_{\alpha\beta}=(1+\kappa)\varrho U_\alpha U_\beta+\kappa\varrho g_{\alpha\beta},$$
here the equation of state for the fluid is the isothermal one: $p=\kappa\varrho$ where $\kappa\in(0,1)$ is a constant, the square of sound speed. $\varrho,p, U$ are the density, pressure and $4$-velocity of the fluid, and $U^\alpha U_\alpha=-1$. The spherically symmetric and continuously self-similar naked singularity solutions of isothermal perfect fluid (for example, those in \cite{O-P90, G-H-J23})) are solutions of \eqref{EinsteinSE} with $\phi\equiv0$. A rough version of the main result of this paper is the following.
\begin{theorem}\label{roughmain}
The spherically symmetric and continuously self-similar naked singularity solutions as constructed in \cite{O-P90, G-H-J23}, as a solution of the whole system \eqref{EinsteinSE}, are unstable to trapped surface formation under $C^{1,\alpha}$ ($\alpha$ depends on the particular naked singularity solution) perturbations of scalar field $\phi$.
\end{theorem}
The precise version is Theorem \ref{main}, presented in the Section \ref{subsection:precise}. Although the naked singularities constructed in \cite{O-P90, G-H-J23} are the only known examples, we remark that our theorem (see precise version below) applies to much wider class of spherically symmetric naked singularity solutions that could be far from being continuously self-similar.

\subsection{The geometry of the past null cone of the singularity} In the proof of naked singularity instability before, the most important feature is the geometry of the past null cone. Suppose that $\Cb_0$ is the past null cone of the singularity $\mathcal{O}$ (see Figure \ref{figure:naked}), in the case of Einstein--scalar field system, $\mathcal{O}$ is a singularity (of $C^1$ solution, see \cite{Chr93}) if and only if the mass ratio $\mu=\frac{2m}{r}\nrightarrow0$ when approaching $\mathcal{O}$ from the past of $\mathcal{O}$. This implies the integral along $\Cb_0$, which is a measurement of the blue-shift of light received by $\mathcal{O}$,
\begin{equation}\label{blueshift}\int_{\Cb_0}\frac{\mu}{1-\mu}\frac{\D r}{r}\end{equation}
is infinite. Christodoulou's proof in \cite{Chr99} makes essential use of this quantity. In our new proof \cite{Liu-Li18}, we observed that the divergence of \eqref{blueshift} implies that the lapse function $\Omega$ of a double null foliation $(\ub,u)$ around $\Cb_0$, with $u=-r$ on $\Cb_0$, tends to zero as approaching $\mathcal{O}$ along $\Cb_0$, which is one of the most important ingredients of our  proof in \cite{Liu-Li18}.

\begin{center}
\begin{figure}
\includegraphics[width=2 in]{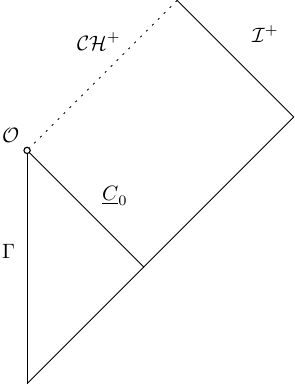}
\caption{\label{figure:naked}}
\end{figure}
\end{center}

We are thus led to see whether other naked singularity solutions we are interested in verify the same condition.  In order to check the behavior of $\Omega$ along $\Cb_0$, we turn to the Raychaudhuri equation along incoming direction, the equation \eqref{Dbhb} in the section below. Together with the gauge choice 
\begin{equation}\label{gaugeonCb0}u=-r\end{equation}
 along $\Cb_0$, we have $\hb=\frac{\partial r}{\partial u}=-1$ and hence \eqref{Dbhb} becomes
\begin{equation}\label{Cb0Omega}
\frac{\partial}{\partial u}\log\Omega\big|_{\Cb_0}=-\frac{1}{4}r\mathbf{Ric}(\partial_u,\partial_u)\big|_{\Cb_0}.
\end{equation}
The question then reduces to see whether the integral of the right hand side on $\Cb_0$ diverges. It turns out that it is related to the so-called strength of the singularity, introduced by Tipler in \cite{Tip77}, which was widely investigated in physics literature. Roughly speaking, a singularity is called a strong curvature singularity (along some causal geodesic approaching it, say $\gamma$), if the volume form generated by Jacobi fields along $\gamma$ tends to zero. A sufficient condition for the singularity being strong along $\gamma$, derived in \cite{Cl-Kr85}, is 
\begin{equation}\label{strongcurvature}\lim_{s\to0}s^2\mathbf{Ric}(\gamma',\gamma')\ne0,\end{equation}
where $s$ is the affine parameter of $\gamma$, and $s\to0$ represents the singularity. It had been checked in \cite{O-P90} that in the naked singularity solutions considered there, condition \eqref{strongcurvature} is true along radial null geodesics emerge from and terminate at $\mathcal{O}$. It can then be seen from equation \eqref{Cb0Omega} that $\Omega\big|_{\Cb_0}\to0$.  This is because if $\Omega\big|_{\Cb_0}\nrightarrow0$, then it has a positive lower bound since it is monotonely decreasing by negativity of the right hand side of \eqref{Cb0Omega}. Then because $\partial_u=\Omega\big|_{\Cb_0}^2\partial_s$, we know $r\approx s$ and $\partial_u\approx\partial_s$. 
Therefore \eqref{strongcurvature} implies that the right hand side of \eqref{Cb0Omega} looks like $\frac{1}{r}$, then by integrating \eqref{Cb0Omega}, we have $\Omega\big|_{\Cb_0}\to0$, a contradiction. It should be noted that $-4\Omega^2\D \ub\D u$ is simply the volume form of the $2$-dimensional quotient spacetime generated by  $\partial_{\ub}$ and $\partial_u$.

In fact, we have a more direct argument to see $\Omega\to0$, and moreover  an accurate rate. Recall that  a spherically symmetric spacetime is called (continuously) self-similar, if there exists a conformal Killing field $S$ such that
$$\mathcal{L}_S \hat{g}=2\hat{g}, Sr=r,$$
and in particular $\mathcal{L}_Sg=2g$, where $\hat{g}$ is the Lorentzian metric induced on the $2$-dimensional quotient spacetime, $r$ is the area radius of the orbit spheres.
Let us assume that, which is true as checked in \cite{O-P90},  $\Cb_0$, the past null cone of the singularity $\mathcal{O}$, is invariant under the flow generated by $S$\footnote{In fact, the meaning of the term ``the past null cone of the singularity $\mathcal{O}$'' should be clarified since $\mathcal{O}$ is not included in the spacetime. However, we can simply take $\Cb_0$ to be the level set of self-similar parameter that is null, if it exists.}. Then on $\Cb_0$ we have $S=r\frac{\partial}{\partial r}=u\frac{\partial}{\partial u}$. The self-similarity also implies that $\mathcal{L}_S\mathbf{Ric}=0$, and hence $S(\mathbf{Ric}(S,S))=0$. Restricted on $\Cb_0$, we will have
\begin{equation}\label{RicCb0}\mathbf{Ric}\left(u\frac{\partial}{\partial u},u\frac{\partial}{\partial u}\right)=2c(=\text{const}\ne0).
\end{equation}
Plugging in \eqref{Cb0Omega}, the right hand side then looks like $\frac{1}{r}$, and then
\begin{equation}\label{ssOmegaCb0}\Omega\big|_{\Cb_0}\approx|u|^c.\end{equation}
Note that this argument applies to all kinds of Einstein--matter field systems since we only use the invariance of the Ricci tensor under self-similarity. It also applies in vacuum if we simply replace the Ricci tensor (which vanishes in vacuum) by the shear tensor (see naked singularity solutions in vacuum constructed in \cite{R-Sh23}). Of course, for a particular matter field system, the variables of matter field should also obey self-similarity. For example, in self-similar Einstein--Euler system with isothermal perfect fluid under consideration, the fluid variables verify on $\Cb_0$ that
\begin{equation}\label{ssfluidCb0}
r^2\varrho=\text{const},  U_u:=g(U,\partial_u)=\text{const}.
\end{equation}
These conditions are of course consistent with \eqref{RicCb0}.

\subsection{The precise statement of the main result and comments}\label{subsection:precise}

Now we are going to give the precise statement we want to prove. The same to before, we consider a double null characteristic problem of the system \eqref{EinsteinSE}, with initial data given on two intersecting null cone, $\Cb_0$ and $C_{u_0}$, where $\Cb_0$ is the past null cone of the singularity $\mathcal{O}$, and $C_{u_0}$ is the outgoing null cone whose intersection with $\Cb_0$ is the sphere $r=-u_0$.

The data on $\Cb_0$ is induced from the naked singularity solutions in \cite{O-P90, G-H-J23}, satisfies \eqref{ssfluidCb0}. In fact, our method is robust and we will allow $r^2\varrho$ and $U_u$ being only bounded above and below by positive constants, together with their $u\partial_u$ derivatives being bounded. In this case, $r^2\mathbf{Ric}(\partial_u,\partial_u)$ on $\Cb_0$ is bounded above and below by positive constants, and hence
\begin{equation}\label{OmegaCb0upperlower}|u|^{\alpha_1}\lesssim \Omega\big|_{\Cb_0}\lesssim |u|^{\alpha_2}\end{equation}
holds with $\alpha_1\geqslant\alpha_2$, instead of \eqref{ssOmegaCb0}. In the proof, we only need the upper bound of $\Omega\big|_{\Cb_0}$. We remark that in our previous works \cite{Liu-Li18, Li-Liu22} on Einstein--scalar field system, only $\Omega\big|_{\Cb_0}\to0$ is needed, no other informations about $\Omega\big|_{\Cb_0}$ and even no informations on the data on $\Cb_0$ are needed. The data on $C_{u_0}$ is spherically symmetric, consists of the fluid variables $(\varrho, U)$ and the scalar field $\phi$, or more precisely, the derivative $\partial_{\ub}\phi$ along $C_{u_0}$.  For we only consider smooth solution and $\phi$ vanishes on $\Cb_0$, $\partial_{\ub}\phi$ is required to be smooth up to $C_{u_0}\cap \Cb_0$ where $\partial_{\ub}\phi=0$. The main result of this paper is the following (also see Figure \ref{figure:perturbations}).

\begin{theorem}\label{main}
Consider double null characteristic problem of the system \eqref{EinsteinSE} in spherical symmetry. Suppose that on the incoming null cone $\Cb_0$, endowed with a parameter $u=-r$, $\phi\equiv0$ and there is an $B\geqslant 1$ such that on $\Cb_0$,
$$ r^2\varrho, (r^2\varrho)^{-1}, |U_u|, |(U_u)^{-1}|, |r^3\partial_u\varrho|, |r\partial_u U_u| \leqslant B.$$
We also assume the Hawking mass $m\big|_{\Cb_0}\geqslant 0$ and all spherical sections of $\Cb_0$ are not trapped\footnote{If there is a trapped section on $\Cb_0$, then the vertex singularity is not naked.}. Then there is a one-parameter family of initial $\phi_t$ (for $t$ sufficiently small) on $C_{u_0}$, endowed with a parameter $\ub\in[0,1]$ such that 
$$r \partial_{\ub}\phi_t\to0, t\to0$$
in $C^{\alpha}$ topology (as a function of $\ub\in[0,1]$), for some $\alpha>0$ depending on $B$, and the maximal future development of such initial data has a closed trapped surface preceding $\mathcal{O}$ for $t\ne0$.

Moreover, the perturbations are generic in the following sense: the set of all data $r\partial_{\ub}\phi$ on $C_{u_0}$ such that there is a closed trapped surface preceding $\mathcal{O}$ in its maximal future development, contains an open set of which zero data (leading to naked singularity solution) is its limit point in $C^\alpha$ topology.
\end{theorem}
\begin{remark}
The non-negativity of the mass $m\big|_{\Cb_0}\geqslant0$ on $\Cb_0$ holds true in naked singularity solutions, which admit a regular central line. In fact, in this double null setting, our proof also applies to the case when $\frac{2m}{r}\big|_{\Cb_0}$ is bounded below by a negative constant, in which case the solution cannot admit a regular central line. 
\end{remark}

\begin{center}
\begin{figure}
\includegraphics[width=2 in]{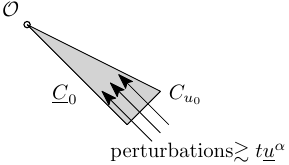}
\caption{\label{figure:perturbations}}
\end{figure}
\end{center}

The proof of this theorem contains the following two main ingredients: 

\begin{itemize}

\item The instability mechanism: This part is similar to the previous work \cite{Liu-Li18}. Roughly speaking, once we can prove that in the region $\frac{\ub}{|u|}\ll1$, all dimensionless variables of the system (such as expansions of null cones $\partial_{\ub}r,\partial_u r$, derivatives of scalar field $|u|\partial_{\ub}\phi, |u|\partial_u\phi$ and fluid variables $|u|^2\varrho, U_u$) are suitably bounded (in spherical symmetry, this can be done by method of characteristic), then by integrating the wave equation \eqref{DbLphi}, that is,
$$\partial_u(r\partial_{\ub}\phi)=-\partial_{\ub}r\partial_{u}\phi$$
 along $u$ direction, we expect that $r\partial_{\ub}\phi$ on $C_{u}$ for $u\in(u_0,0)$ is as large as its value on the initial null cone $C_{u_0}$. Then the Raychaudhuri equation along $C_u$, equation \eqref{Dh2} written in the following form
 $$\partial_{\ub}(\Omega^{-2}\partial_{\ub}r)\leqslant -r\Omega^{-2}(\partial_{\ub}\phi)^2,$$
tells us that the right hand side becomes sufficiently large when $u\to0$ because $\Omega\to0$. So $\partial_{\ub} r$ can become negative when integrating along $\ub$ when $|u|$ is sufficiently small. A slight difference from \cite{Liu-Li18} is that we don't have a criterion of trapped surface formation like that in \cite{Chr91}, which is crucial in Christodoulou's original proof \cite{Chr99}. Nevertheless, we can still prove trapped surface formation in a direct way.

\item The estimates for fluid variables: According to our setting (the perturbations are added exterior to the past null cone of $\mathcal{O}$), perturbations from fluid cannot produce trapped surface around the singularity since it has slower wave speed than light, then the past sound cone of points near $\mathcal{O}$ will not intersect $C_{u_0}$. On the other hand, since the background geometry changes a lot (trapped surfaces form due to perturbations of scalar field), we must carefully show that, in the region we are working in, the fluid variables still behave well (for example, shocks will not form) and obey the self-similar estimates which are needed in the previous ingredient. The main reason why this works is that solving the Euler equations in the region $\frac{\ub}{|u|}\ll1$ is a local problem (we will prove that both characteristic curves have length $\lesssim \Omega^2\ub$), while solving Einstein--scalar field equations is semi-global. 
\end{itemize}

We close this section by some final remarks.

\begin{remark}
Though the perturbations are measured in $C^{1,\alpha}$, we only consider smooth solutions to avoid details in establishing local--wellposedness of weak solution of the system \eqref{EinsteinSE}, which is possible because we don't need $\partial_{\ub}^2\phi$ to close the apriori estimates. Local wellposedness of smooth solutions of spherically symmetric Einstein--Euler system in double null setting was established in \cite{G-H-J23}. If passing to weak solutions, the perturbations looks like $r\partial_{\ub}\phi_t\gtrsim t\ub^{\alpha}$ when $\ub\to0^+$ and the apparent horizon emerging from the singularity $\mathcal{O}$ can be found, just as the original Einstein--scalar field system.  Smooth perturbations are obtained by cutting off.
\end{remark}

\begin{remark}
 In the current paper, the reason why we can establish $C^{1,\alpha}$ instability (which is better than all previous related works) is rooted in the quantitative behavior \eqref{OmegaCb0upperlower}, in which only the upper bound is needed. But it should be noted that the H\"older index depends on particular background naked singularity solution. Finding the optimal H\"older index is an interesting problem. See more detailed discussions in \cite{JS24}. 
\end{remark}
\begin{remark}
We expect that our method apply to most matter fields that have a slower sound speed than light, showing that the corresponding naked singularity solutions, when the singularities are strong, are unstable to trapped surface formation under gravitational perturbations. This will support a conjecture formulated by Newman \cite{Newman86} that naked singularities appearing in gravitational collapse are in some sense gravitationally weak. Note that the naked singularity solutions of spherically symmetric inhomogeneous dust collapse ($\kappa=0$) constructed by Christodoulou \cite{Chr84} has a weak naked singularity along radial past and future null geodesics \cite{Newman86cqg}. There exist some exceptional choices of initial configurations of spherically symmetric inhomogeneous dust collapse such that the resulting naked singularities are strong along past and future null geodesics (see for example \cite{S-J96}), in which case our argument can also apply. We can also consider the case of stiff fluid ($\kappa=1$), even though we don't have a precise example of naked singularity solution in this case. But in this case, the existence of the fluid part becomes semi-global and then it raises new difficulties.  Moreover, because the sound speed of stiff fluid is the same to the light speed, one may hope to show instability without external scalar field.
\end{remark}

\subsection*{Acknowledgement}
This work is supported by National Key R\&D Program of China (No. 2022YFA1005400), NSFC (12326602, 12141106).

\section{Double null coordinate and equations}

\subsection{Double null coordinate}

In a spherically symmetric spacetime,  let us introduce the double null coordinate $(\ub,u)$, where $\ub$, $u$ are optical functions, which means that their level sets $\Cb_{\ub}$ and $C_u$ are incoming and outgoing null cones respectively, invariant under the $SO(3)$ action representing spherical symmetry. We denote $S_{\ub,u}=C_u\cap\Cb_{\ub}$,  which is an orbit sphere. In the quotient spacetime, $\Cb_{\ub}$ and $C_u$ are then incoming and outgoing null lines respectively. We  denote
\begin{align*}
L=\frac{\partial}{\partial \ub},\ \Lb=\frac{\partial}{\partial u},
\end{align*}
to be the coordinate vector fields, and the lapse function $\Omega$ is defined by
\begin{align*}
-2\Omega^2=g(L,\Lb).
\end{align*}
The  metric then takes the following form
\begin{align*}
-2\Omega^2(\D\ub\otimes\D u+\D u\otimes\D\ub)+r^2\D\sigma_{\mathbb{S}^2}
\end{align*}
where the area radius function $r=r(\ub,u)$ is defined by
\begin{align*}
\text{Area}(S_{\ub,u})=4\pi r^2,
\end{align*}
and $\D\sigma_{\mathbb{S}^2}$ is the standard metric of the unit sphere.

\subsection{Equations} The estimates are done using the null structure equations. In spherical symmetric case,  the non-vanishing components of connections are 
\begin{align*}
h=\Omega^{-2}D r,\ \hb=\Db r, \omega=D\log\Omega, \omegab=\Db\log\Omega, 
\end{align*}
where $D$ and $\Db$ are the restrictions on the orbit spheres of the Lie derivatives along $L$ and $\Lb$. When applying on functions, $D$ and $\Db$ are simply the ordinary derivatives. Here $h$ and $\hb$ are the null expansions relative to the normalized pair of null vectors $\Omega^{-2}L$, $\Lb$. Then null structure equations are\footnote{These equations can be derived as in \cite{Chr08}, in which the equations without any symmetries in vacuum are obtained. One only needs to carefully bring back the Ricci tensor. Readers can also see \cite{Chr91} in which the notations are slightly different. }
\begin{align}
\label{Dh}Dh=&-\frac{1}{2}r\Omega^{-2}\mathbf{Ric}_{LL},\\
\label{Dbhb}\Db(\Omega^{-2}\hb)=&-\frac{1}{2}r\Omega^{-2}\mathbf{Ric}_{\Lb\Lb},\\
\label{DbhDhb}\Db(\Omega^2h)=D\hb=&-\frac{\Omega^2(1+h\hb)-\frac{1}{2}r^2\Omega^2\trs\mathbf{Ric}}{r},\\
\label{Dbomega}\Db\omega=D\omegab=&\frac{\Omega^2(1+h\hb)}{r^2}-\frac{1}{2}(\Omega^2\trs\mathbf{Ric}+\mathbf{Ric}_{L\Lb}).
\end{align}
Here $\trs\mathbf{Ric}$ represents the trace of $\mathbf{Ric}$ taken on the orbit spheres.

For the Einstein-scalar field-Euler system \eqref{EinsteinSE}, the energy momentum tensor reads
$$\mathbf{T}=\mathbf{T}_{\alpha\beta}^{sc}+\mathbf{T}_{\alpha\beta}^{fl}$$
where
$$\mathbf{T}_{\alpha\beta}^{sc}=\partial_\alpha\phi\partial_\beta\phi-\frac{1}{2}g_{\alpha\beta}g^{\mu\nu}\partial_\mu\phi\partial_\nu\phi,$$
and
$$\mathbf{T}_{\alpha\beta}^{fl}=(\varrho+p) U_\alpha U_\beta+pg_{\alpha\beta}=(1+\kappa)\varrho U_\alpha U_\beta+\kappa\varrho g_{\alpha\beta}.$$
From \eqref{EinsteinSE}, we then have
$$\mathbf{Ric}_{\alpha\beta}=2\partial_\alpha\phi\partial_\beta\phi+2(1+\kappa)\varrho U_{\alpha}U_{\beta}+(1-\kappa)\varrho g_{\alpha\beta}.$$
Plugging in the above expressions of the Ricci tensor, equations \eqref{Dh}-\eqref{Dbomega} become:
\begin{align}
\label{Dh2}Dh=&-r\Omega^{-2}(L\phi)^2-r\Omega^{-2}(1+\kappa)\varrho(U_L)^2,\\
\label{Dbhb2}\Db(\Omega^{-2}\hb)=&-r\Omega^{-2}(\Lb\phi)^2-r\Omega^{-2}(1+\kappa)\varrho(U_{\Lb})^2,\\
\label{DbhDhb2}\Db (\Omega^2 h)=D\hb=&-\frac{\Omega^2(1+h\hb-r^2(1-\kappa)\varrho)}{r},\\
\label{Dbomega2}\Db \omega=D\omegab=&\frac{\Omega^2(1+h\hb)}{r^2}-L\phi\Lb\phi-(1+\kappa)\varrho.
\end{align}

We turn to the equations satisfied the matter fields. The scalar field $\phi$ obeys the wave equation
$$\Box_g\phi=0,$$
and the fluid variables $(\varrho, U)$ obey the relativistic Euler equations
$$ U_\alpha\nabla^\alpha\varrho+(1+\kappa)\varrho\nabla^\alpha U_\alpha=0,$$
$$(1+\kappa)\varrho U_\alpha\nabla^\alpha U_\beta+\kappa\Pi_\beta^\alpha\nabla_\alpha\varrho=0,$$
with $U_\alpha U^\alpha=-1$, where $$\Pi_{\alpha\beta}=g_{\alpha\beta}+(1-\kappa)U_\alpha U_\beta$$ is the acoustical metric.

We should decompose the wave and Euler equations in double null coordinate. Let us introduce
\begin{align*}
L\phi=\frac{\partial}{\partial\ub}\phi,\ \Lb\phi=\frac{\partial}{\partial u}\phi,
\end{align*} 
then the wave equation can be written in the following form:
\begin{align}
\label{DbLphi}\Db(rL\phi)=&-\Omega^2h\Lb\phi,\\
\label{DLbphi}D(r\Lb\phi)=&-\hb L\phi.
\end{align}
These two equations are in fact the same equation.  For fluid variables, we denote
$$U_{\Lb}=g(U,\Lb), U_{L'}=g(U,L'),$$
where $L'=\Omega^{-2}L$ is such that $g(\Lb,L')=-2$. We then have
$$U=-\frac{1}{2}U_{L'} \Lb-\frac{1}{2}U_{\Lb}L'.$$
We have $U_{\Lb}<0, U_{L'}<0$ and $U_{\Lb}U_{L'}=1$ since $U$ is a future directed unit timelike vector.

We should write the relativistic Euler equations in the form allowing us to use the method of characteristics. 
\begin{lemma} A regular solution $(\varrho, U)$ of the relativistic Euler equations verify the following equations whenever $\varrho>0, U_{\Lb}, U_{L'}<0$:
\begin{align}\label{euler1}U_+\log\left(\varrho^{\frac{2\sqrt{\kappa}}{1+\kappa}} U_{\Lb}^2\right)=-2\omegab U_{L'}\sqrt{\frac{1-\sqrt{\kappa}}{1+\sqrt{\kappa}}}+\frac{2}{r}\sqrt{\frac{\kappa}{1-\kappa}}(U_{L'}\hb+U_{\Lb}h),\end{align}
	\begin{align}\label{euler2}U_-\log\left(\varrho^{\frac{2\sqrt{\kappa}}{1+\kappa}} U_{L'}^2\right)=2\omegab U_{L'}\sqrt{\frac{1+\sqrt{\kappa}}{1-\sqrt{\kappa}}}-\frac{2}{r}\sqrt{\frac{\kappa}{1-\kappa}}(U_{L'}\hb+U_{\Lb}h).\end{align}	where $U_+, U_-$ are the outgoing and incoming future directed acoustical-null vectors ($\Pi_{\alpha\beta}U_{\pm}^{\alpha}U_{\pm}^\beta=0$) of unit length:
	\begin{align}\label{acousticalnull}U_\pm=-\frac{1}{2}\sqrt{\frac{1\mp\sqrt{\kappa}}{1\pm\sqrt{\kappa}}} U_{L'}\Lb-\frac{1}{2}\sqrt{\frac{1\pm\sqrt{\kappa}}{1\mp\sqrt{\kappa}}} U_{\Lb}L'.\end{align}
\end{lemma}
\begin{remark}
The crucial fact is that the right hand sides contain no derivatives of $(\varrho, U)$. The expressions \eqref{acousticalnull} of $U_{\pm}$ can be seen in the course of the proof, or directly obtained by the method of undetermined coefficients by $\Pi_{\alpha\beta}U_{\pm}^{\alpha}U_{\pm}^\beta=0$ and $U^\alpha_{\pm}U_{\pm\alpha}=-1$. \end{remark}
\begin{proof}
These equations can also be found in \cite{G-H-J23} written in different notations. For completeness, we outline the computations here. The Euler equations are simply the conservation law of the energy momentum tensor:
	$$\nabla^\alpha\mathbf{T}^{fl}_{\alpha\beta}=\nabla^\alpha((1+\kappa)\varrho U_\alpha U_\beta+\kappa\varrho g_{\alpha\beta})=0. $$
Using the frame $(L', \Lb, e^A)$ where $e^A, A=1,2$ are spherical, its $L', \Lb$ components are
\begin{equation}\label{euler0-}\begin{split}\nabla_{L'}\mathbf{T}^{fl}_{\Lb L'}+\nabla_{\Lb}\mathbf{T}^{fl}_{L' L'}-2\gs^{AB}\nabla_{A}\mathbf{T}^{fl}_{BL'}=0,\\
\nabla_{L'}\mathbf{T}^{fl}_{\Lb \Lb}+\nabla_{\Lb}\mathbf{T}^{fl}_{L' \Lb}-2\gs^{AB}\nabla_{A}\mathbf{T}^{fl}_{B\Lb}=0,\end{split}\end{equation}
where $\gs$ is the round metric of radius $r$ induced on spherical sections. By direct checking, we have the following formulas of connections:
$$\nabla_{L'}L'=0, \nabla_{L'}\Lb=0, \nabla_{\Lb}L'=-2\omegab L', \nabla_{\Lb}\Lb=2\omegab\Lb,$$
$$\nabla_{e_A}L'=\frac{h}{r}e_A, \nabla_{e_A}\Lb=\frac{\hb}{r}e_A,$$
$$\nabla_{e_A}e_B=\nablas_{e_A}e_B+\gs_{AB}\left(\frac{h}{2r}\Lb+\frac{\hb}{2r}L'\right),$$
and the only non-vanishing components of $\mathbf{T}^{fl}$:
$$\mathbf{T}^{fl}_{L'L'}=(1+\kappa)\varrho U_{L'}^2,\mathbf{T}^{fl}_{\Lb\Lb}=(1+\kappa)\varrho U_{\Lb}^2,$$
$$\mathbf{T}^{fl}_{L'\Lb}=(1-\kappa)\varrho, \mathbf{T}^{fl}_{AB}=\kappa \varrho \gs_{AB},$$
where we use $U_{\Lb}U_{L'}=1$ in evaluating $\mathbf{T}^{fl}_{L'\Lb'}$. Then \eqref{euler0-} becomes
	\begin{equation*}
	\begin{split}\frac{1-\kappa}{1+\kappa}L'\varrho+\Lb(\varrho U_{L'}^2)+4\omegab\varrho U_{L'}^2+\frac{2h}{r}\varrho+\frac{2\hb}{r}\varrho U_{L'}^2=0,\\
	L'(\varrho U_{\Lb}^2)+\frac{1-\kappa}{1+\kappa}\Lb\varrho+\frac{2\hb}{r}\varrho+\frac{2h}{r}\varrho U_{\Lb}^2=0.
	\end{split}
	\end{equation*}
	Dividing the first by $\varrho U_{L'}$ and the second by $\varrho U_{\Lb}$ we have (using again $U_{\Lb}U_{L'}=1$)
		\begin{equation}\label{euler0}
	\begin{split}\frac{1}{s-1}\frac{1-\kappa}{1+\kappa}U_{\Lb} L'\log\varrho^{s-1}+U_{L'}\Lb\log(\varrho U_{L'}^2)+4\omegab U_{L'}+\frac{2h}{r}U_{\Lb}+\frac{2\hb}{r} U_{L'}=0,\\
	U_{\Lb}L'\log(\varrho U_{\Lb}^2)+\frac{1}{-s-1}\frac{1-\kappa}{1+\kappa}U_{L'}\Lb\log\varrho^{-s-1}+\frac{2\hb}{r}U_{L'}+\frac{2h}{r}U_{\Lb}=0.
	\end{split}
	\end{equation}
	where $s$ is to be determined. We then multiply the second equation by $\frac{1}{s-1}\frac{1-\kappa}{1+\kappa}$ and then compare the coefficient of the second term, that is,
	\begin{equation}\label{equationofs}\frac{1}{s-1}\frac{1}{-s-1}\left(\frac{1-\kappa}{1+\kappa}\right)^2=1.\end{equation}
	Then we can sum up two equations to obtain equations for $\varrho^s U_{\Lb}^2$ and $\varrho^s U_{L'}^2$. The positive root of \eqref{equationofs} is simply $s=\frac{2\sqrt{\kappa}}{1+\kappa}$, as indicated in \eqref{euler1}, \eqref{euler2}. The negative root of \eqref{equationofs} leads to the same equations. 
\end{proof}

\section{A priori estimates}

We begin to prove Theorem \ref{main}. In this section, we are going to derive a priori estimates.  Recall that we consider double null characteristic problem of Einstein--scalar field-Euler system \eqref{EinsteinSE} in spherical symmetry. The incoming cone $\Cb_0$  is endowed with the parameter $u=-r$, where $u$ is one of the optical functions in the spacetime we consider, and also $\phi\equiv0$ on $\Cb_0$. We denote
$$\Omega_0(u)=\Omega(0,u), h_0(u)=h(0,u)$$
to be the restriction of $\Omega$ and $h$ on $\Cb_0$. Note that since $\Omega_0$ is decreasing when $u\to0$ by \eqref{Cb0Omega} and energy condition, by rescaling $\ub$, we can set $\Omega_0(u_0)\leqslant 1$ and then $\Omega_0(u)\leqslant 1$ for $u\in [u_0,0)$. Note also that we have $h_0\leqslant 1$, which is equivalent to the positivity of the Hawking mass $m\big|_{\Cb_0}=\frac{r}{2}(1-h_0)$ (on $\Cb_0$ we have $\hb=-1$) and $h_0\geqslant 0$ since $\Cb_0$ has no closed trapped surfaces on it. 
In the following, we use
$$A\lesssim B$$
to denote $A\leqslant cB$ for some constant $c>0$, depending only on $\kappa$.
\begin{theorem}\label{apriori1}

There exists a constant $C>0$ depending on $\kappa$ such that the following statement is true. Let $a,b\geqslant 1$ be such that 
$$\sup_{\Cb_0\cup C_{u_0}}\big\{|u|^2\varrho, (|u|^2\varrho)^{-1}\big\}+\sup_{C_{u_0}}\{|rL\phi|, |r\omega|\}\leqslant a, $$
$$\sup_{\Cb_0\cup C_{u_0}}\big\{|U_{\Lb}|, |U_{L'}|\big\}\leqslant b.$$
 Then for any $\delta>0, u_1\in (u_0,0)$ such that
\begin{align}\label{smallness1}C\delta|u_1|^{-1}a^{2}b^{\frac{(1+\sqrt{\kappa})^2}{\sqrt{\kappa}}}\leqslant 1,\end{align}
the following estimates holds for $0\leqslant\ub\leqslant\delta, u_0\leqslant u\leqslant u_1<0$, 
\begin{align}\label{geometry}\frac{1}{2}\Omega_0(u)\leqslant\Omega\leqslant 2\Omega_0(u),  \frac{1}{2}|u|\leqslant r\leqslant 2|u|,\end{align}
\begin{align}
\label{h-h0}|h-h_0|\lesssim& \delta|u|^{-1}\Omega_0^{-2} a^2b^{\frac{(1+\sqrt{\kappa})^2}{\sqrt{\kappa}}},\\
\label{hb+1}|\hb+1|\lesssim& \delta|u|^{-1}ab^{\frac{1+\kappa}{\sqrt{\kappa}}},\\
\label{omega}|u||\omega|\lesssim& ab^{\frac{1+\kappa}{\sqrt{\kappa}}},\\
\label{omegab}|u||\omegab|\lesssim& ab^2,
\end{align}

\begin{align}\label{Lphi}|rL\phi|\lesssim& a,\\
\label{Lbphi}|r\Lb\phi|\lesssim& \delta|u|^{-1}a,\\
\label{varrho}|u|^2|\varrho|\lesssim& ab^\frac{1+\kappa}{\sqrt{\kappa}},\\
\label{U}|U_{\Lb}|, |U_{L'}|\lesssim& a^\frac{\sqrt{\kappa}}{1+\kappa}b,
\end{align}
\end{theorem}
\begin{proof}
It suffices to prove the estimates under the bootstrap assumptions
\begin{align}
\label{h-h0b}|h-h_0|\lesssim& C^{\varepsilon}\delta|u|^{-1}\Omega_0^{-2} a^2b^{\frac{(1+\sqrt{\kappa})^2}{\sqrt{\kappa}}},\\
\label{hb+1b}|\hb+1|\lesssim& C^{\varepsilon}\delta|u|^{-1}ab^{\frac{1+\kappa}{\sqrt{\kappa}}},\\
\label{omega-b}|u||\omega|\lesssim& C^{\varepsilon}ab^{\frac{1+\kappa}{\sqrt{\kappa}}},\\
\label{omegabb}|u||\omegab|\lesssim& C^{\varepsilon}ab^{2},
\end{align}

\begin{align}
\label{Lbphib}|r\Lb\phi|\lesssim& C^{\varepsilon}\delta|u|^{-1}a,\\
\label{Ub}|U_{\Lb}|, |U_{L'}|\lesssim& C^{\varepsilon} a^\frac{\sqrt{\kappa}}{1+\kappa}b,
\end{align}
where $\varepsilon>0$ is small ($\varepsilon=\frac{1}{4}$ is sufficient) and $C\geqslant 1$ is sufficiently large and to be determined.

From the equation $D\log\Omega=\omega$, using  \eqref{smallness1} and the bound \eqref{omega-b}, we have
$$|\log\Omega-\log\Omega_0|\leqslant \int_0^\delta |\omega|\D\ub\lesssim C^\varepsilon\delta|u|^{-1}ab^{\frac{1+\kappa}{\sqrt{\kappa}}}\leqslant C^{\varepsilon-1}.$$
If $C$ is sufficiently large, $|\log\Omega-\log\Omega_0|\leqslant \log 2$ and the estimate \eqref{geometry} for $\Omega$ follows.  Here we use the fact that $b\geqslant1$ and $(1+\sqrt{\kappa})^2\geqslant 1+\kappa$.

To estimate $r$, note that if $C$ is sufficiently large,
\begin{align}\label{h}|\Omega^2 h|\lesssim\Omega_0^2(|h_0|+C^{\varepsilon}\delta|u|^{-1}\Omega_0^{-2} a^2b^{\frac{(1+\sqrt{\kappa})^2}{\sqrt{\kappa}}})\lesssim 1,\end{align}
where we have used \eqref{smallness1}, \eqref{geometry} for $\Omega$ and $\Omega_0, |h_0|\leqslant 1$. Then from $Dr=\Omega^2h$, we have
$$|r-|u||\leqslant\int_0^\delta|\Omega^2h|\D\ub\lesssim \delta \leqslant C^{-1}|u|.$$
Then if $C$ is sufficiently large, the estimate \eqref{geometry} for $r$ holds.

Next we estimate the matter fields. From the wave equation $\Db(rL\phi)=-\Omega^2 h\Lb\phi$, using \eqref{smallness1}, \eqref{Lbphib} and \eqref{h},
\begin{align*}|rL\phi|\leqslant&|rL\phi\big|_{C_{u_0}}|+\int_{u_0}^u |\Omega^2h||\Lb\phi|\D u'\\
\lesssim &a+C^\varepsilon a\int_{u_0}^u\delta|u'|^{-2}a\D u'\leqslant (1+C^{\varepsilon-1}) a\lesssim a,\end{align*}
and from the other version of the wave equation $D(r\Lb\phi)=-\hb L\phi$, using the above estimate, 
$$|r\Lb\phi|\leqslant \int_0^\delta |\hb L\phi|\D\ub\lesssim\delta |u|^{-1}a,$$
where we have used 
\begin{align}\label{hb}|\hb|\leqslant |\hb+1|+1\lesssim C^\varepsilon\delta|u|^{-1}ab^{\frac{1+\kappa}{\sqrt{\kappa}}} +1\lesssim 1. \end{align}
This proves estimates \eqref{Lphi} and \eqref{Lbphi}.

Next we turn to the estimates for the fluid variables. For each $(\ub,u)\in [0,\delta]\times[u_0,u_1]$, let $\tau\mapsto\gamma_\pm(\tau;\ub,u)$ be the integral curve of $U_{\pm}$ passing through $(\ub,u)$, with $\gamma_\pm(0;\ub,u)\in C_{u_0}\cup \Cb_0$ (The existence of such a point can be argued as follows: For an inextendible timelike curve with unit tangent, it must intersect the past boundary of a global hyperbolic spacetime). Note that for different $(\ub,u)$ (those on the same integral curve), $\gamma_{\pm}(\tau; \ub,u)$ can be the same. We denote $\tau_{\pm}(\ub,u)$ be the parameter of $\gamma_{\pm}(\tau;\ub,u)$ at $(\ub,u)$. 

In order to estimate the fluid variables, we need to estimate the length of $\gamma_{\pm}$ in the spacetime region that we consider, and the variation of $|u|$ and $\Omega_0(u)$ along $\gamma_{\pm}$. We begin by assuming that, for any $(\ub,u)$,
\begin{equation}\label{ucompare}0\leqslant \frac{|u(\gamma_{\pm}(0;\ub,u))|}{|u|}-1\leqslant (ab^2)^{-1}\log 2,\end{equation}
which is an estimate of the variation of the function $|u|$ along the integral curves $\gamma_{\pm}$ of $U_{\pm}$.
The the equation \eqref{Dbhb2} restricted on $\Cb_0$ (where $u=-r$ and hence $\hb=-1$) reads
\begin{equation}\label{DbhbCb0}\frac{\partial}{\partial u}\log\Omega_0=-\frac{1}{2}|u|(1+\kappa)\varrho(U_{\Lb})^2.\end{equation}
Particularly, we can see again this implies that $\Omega_0$ is monotonically decreasing along $\Cb_0$. On $\Cb_0$ we have $\varrho(U_{\Lb})^2\leqslant |u|^{-2}ab^2$, then by integrating it from $u(\gamma_{\pm}(0;\ub,u))$ to $u$, we have
\begin{align*}\log\frac{\Omega_0(u(\gamma_{\pm}(0;\ub,u)))}{\Omega_0(u)}=&\int_{u(\gamma_{\pm}(0;\ub,u))}^{u}\frac{1}{2}|u'|(1+\kappa)\varrho(U_{\Lb})^2\big|_{\Cb_0}\D u'\\\leqslant& ab^2 \log\frac{|u(\gamma_{\pm}(0;\ub,u))|}{|u|}\\
\leqslant& ab^2\left(\frac{|u(\gamma_{\pm}(0;\ub,u))|}{|u|}-1\right)\\
\leqslant&\log 2,\end{align*}
where the last inequality is by \eqref{ucompare}. This implies that, for $u_0\leqslant u\leqslant u_1$, 
\begin{align*}1\leqslant \frac{\Omega_0(u(\gamma_{\pm}(0;\ub,u)))}{\Omega_0(u)}\leqslant 2.\end{align*}
By the monotonicity of $\Omega_0$,  the $\gamma_{\pm}(0;\ub,u)$ above can be replaced by $\gamma_{\pm}(\tau;\ub,u)$ for any $\tau\in[0,\tau_{\pm}(\ub,u)]$, then using \eqref{geometry} we have
\begin{align}\label{Omegaalonggamma}\frac{1}{2}\leqslant \frac{\Omega(\gamma_{\pm}(\tau;\ub,u))}{\Omega_0(u)}\leqslant 4,\end{align}
for $0\leqslant \tau\leqslant \tau_{\pm}(\ub,u)$. This is the variation of $\Omega$ along $\gamma_{\pm}$.

Then we are going to estimate $\tau_{\pm}(\ub,u)$, representing the length of $\gamma_{\pm}$.
Note that
$$U_{\pm}(\ub)=-\frac{1}{2}\sqrt{\frac{1\pm\sqrt{\kappa}}{1\mp\sqrt{\kappa}}}\Omega^{-2}U_{\Lb},$$
 by integrating the right hand side along $\gamma_{\pm}$ from $0$ to $\tau_{\pm}(\ub,u)$, we have
\begin{align*}\ub=&\int_0^{\tau_{\pm}(\ub,u)}-\frac{1}{2}\sqrt{\frac{1\pm\sqrt{\kappa}}{1\mp\sqrt{\kappa}}}\Omega^{-2}U_{\Lb}\Big|_{\gamma_{\pm}(\tau;\ub,u)}\D \tau'\\
\gtrsim& C^{-\varepsilon}\Omega_0^{-2}(u)(a^\frac{\sqrt{\kappa}}{1+\kappa}b)^{-1}\tau_\pm(\ub,u),
\end{align*}
where we have used \eqref{Ub} and  and \eqref{Omegaalonggamma}. This implies that 
\begin{align}\label{taub}\tau_\pm(\ub,u)\lesssim  C^{\varepsilon}\Omega_0^2(u)\ub a^\frac{\sqrt{\kappa}}{1+\kappa}b.\end{align}
Now we are going to retrieve the estimate \eqref{ucompare}. From 
$$U_\pm(u)=-\frac{1}{2}\sqrt{\frac{1\mp\sqrt{\kappa}}{1\pm\sqrt{\kappa}}} U_{L'},$$
 the difference of the values of $u$ at the starting point and ending point of $\gamma_\pm$ is estimated by
\begin{align*}|u-u(\gamma_{\pm}(0;\ub,u))|=&\int_0^{\tau_\pm(\ub,u)}-\frac{1}{2}\sqrt{\frac{1\mp\sqrt{\kappa}}{1\pm\sqrt{\kappa}}} U_{L'}\Big|_{\gamma_{\pm}(\tau;\ub,u)}\D\tau'\\\lesssim &C^{2\varepsilon} \Omega_0^2(u)\ub (a^\frac{\sqrt{\kappa}}{1+\kappa}b)^2 \\\leqslant &C^{2\varepsilon-1}|u|(ab^2)^{-1},\end{align*}
where we have used \eqref{Ub} and \eqref{taub} in the first inequality and \eqref{smallness1} in the second inequality (Note that $a,b\geqslant1$, $\frac{2\sqrt{\kappa}}{1+\kappa}\leqslant1$ and $\frac{(1+\sqrt{\kappa})^2}{\sqrt{\kappa}}\geqslant 4$). We also use $\Omega_0\leqslant 1$. If $C$ is chosen sufficiently large, we then have, for $u_0\leqslant u\leqslant u_1$,
$$|u-u(\gamma_{\pm}(0;\ub,u))|\leqslant |u| (ab^2)^{-1}\frac{1}{2}\log 2$$
which improves the estimate \eqref{ucompare}. This implies that \eqref{ucompare} is in fact true for each $(\ub,u)$ with $\log 2$ replaced by $\frac{1}{2}\log2$ and hence the subsequent estimates are also true.  Moreover, we have
\begin{align}\label{ucompare2}1\leqslant \frac{|u(\gamma_{\pm}(\tau;\ub,u))|}{|u|}\leqslant 2,\end{align}
for $0\leqslant\tau\leqslant\tau_{\pm}(\ub,u)$. The estimates \eqref{Omegaalonggamma} and \eqref{ucompare2} together tells us that when integrating along $\gamma_{\pm}$, we can simply replace $|u(\gamma_{\pm}(\tau;\ub,u)|$ and $\Omega(\gamma_{\pm}(\tau;\ub,u))$ in the integrants by their values at the end points $(\ub,u)$.

We introduce an additional bootstrap assumption
\begin{align}\label{varrhoinverseb}(|u|^2|\varrho|)^{-1}\lesssim C^{\varepsilon}ab^\frac{1+\kappa}{\sqrt{\kappa}},\end{align}
for $0\leqslant \ub\leqslant \delta$, $u_0\leqslant u\leqslant u_1<0$, which guarantees that $\varrho>0$ throughout the spacetime region we consider. Also, since $U_{\Lb}, U_{L'}<0$, we can refer to the Euler equations \eqref{euler1}, \eqref{euler2}, denoting the right hand sides by $R_{\pm}$:
\begin{align}\label{euler12}U_+\log\left(\varrho^{\frac{2\sqrt{\kappa}}{1+\kappa}} U_{\Lb}^2\right)=-2\omegab U_{L'}\sqrt{\frac{1-\sqrt{\kappa}}{1+\sqrt{\kappa}}}+\frac{2}{r}\sqrt{\frac{\kappa}{1-\kappa}}(U_{L'}\hb+U_{\Lb}h)\triangleq R_+,\end{align}
	\begin{align}\label{euler22}U_-\log\left(\varrho^{\frac{2\sqrt{\kappa}}{1+\kappa}} U_{L'}^2\right)=2\omegab U_{L'}\sqrt{\frac{1+\sqrt{\kappa}}{1-\sqrt{\kappa}}}-\frac{2}{r}\sqrt{\frac{\kappa}{1-\kappa}}(U_{L'}\hb+U_{\Lb}h)\triangleq R_-.\end{align}
	Using \eqref{omegabb}, \eqref{Ub} and \eqref{geometry}, \eqref{h}, \eqref{hb}, the right hand sides of the above two equations are bounded by 
	\begin{align*}|R_+|,|R_-|\lesssim& C^{2\varepsilon} |u|^{-1}ab^2\cdot a^{\frac{\sqrt{\kappa}}{1+\kappa}}b+C^{\varepsilon}|u|^{-1}a^{\frac{\sqrt{\kappa}}{1+\kappa}}b\cdot\Omega_0^{-2}(u),\\
	\lesssim& C^{2\varepsilon} \Omega_0^{-2}(u)|u|^{-1}ab^2\cdot a^{\frac{\sqrt{\kappa}}{1+\kappa}}b.\end{align*}
	Using  \eqref{taub} and \eqref{smallness1}, we find
	$$\int_0^{\tau_{\pm}(\ub,u)}|R_{\pm}|\D\tau'\lesssim  C^{2\varepsilon} \Omega_0^{-2}(u)|u|^{-1}ab^2\cdot a^{\frac{\sqrt{\kappa}}{1+\kappa}}b\cdot C^{\varepsilon}\Omega_0^2(u)\ub a^\frac{\sqrt{\kappa}}{1+\kappa}b\leqslant C^{3\varepsilon-1}.$$
	By choosing $C$ sufficiently large, we have
	$$\int_0^{\tau_{\pm}(\ub,u)}|R_{\pm}|\D\tau'\leqslant \log2.$$
By integrating the Euler equations, we have
	$$\left|\log\left(\frac{\varrho^{\frac{2\sqrt{\kappa}}{1+\kappa}} U_{\Lb}^2(\ub,u)}{\varrho^{\frac{2\sqrt{\kappa}}{1+\kappa}} U_{\Lb}^2(\gamma_{+}(0;\ub,u))}\right)\right|\leqslant \int_0^{\tau_+(\ub,u)}|R_+|\D\tau'\leqslant \log 2$$
	and
	$$\left|\log\left(\frac{\varrho^{\frac{2\sqrt{\kappa}}{1+\kappa}} U_{L'}^2(\ub,u)}{\varrho^{\frac{2\sqrt{\kappa}}{1+\kappa}} U_{L'}^2(\gamma_{-}(0;\ub,u))}\right)\right|\leqslant \int_0^{\tau_-(\ub,u)}|R_-|\D\tau'\leqslant \log 2,$$
which implies that
$$\frac{1}{2}\varrho^{\frac{2\sqrt{\kappa}}{1+\kappa}} U_{\Lb}^2(\gamma_{+}(0;\ub,u))\leqslant\varrho^{\frac{2\sqrt{\kappa}}{1+\kappa}} U_{\Lb}^2(\ub,u)\leqslant 2\varrho^{\frac{2\sqrt{\kappa}}{1+\kappa}} U_{\Lb}^2(\gamma_{+}(0;\ub,u))$$
and
$$\frac{1}{2}\varrho^{\frac{2\sqrt{\kappa}}{1+\kappa}} U_{L'}^2(\gamma_{-}(0;\ub,u))\leqslant\varrho^{\frac{2\sqrt{\kappa}}{1+\kappa}} U_{L'}^2(\ub,u)\leqslant 2\varrho^{\frac{2\sqrt{\kappa}}{1+\kappa}} U_{L'}^2(\gamma_{-}(0;\ub,u)).$$
By these two relations, we have
\begin{align*}\varrho=&\left(\varrho^{\frac{2\sqrt{\kappa}}{1+\kappa}} U_{\Lb}^2\cdot \varrho^{\frac{2\sqrt{\kappa}}{1+\kappa}} U_{L'}^2\right)^{\frac{1+\kappa}{4\sqrt{\kappa}}}\\
\lesssim&\left(2\varrho^{\frac{2\sqrt{\kappa}}{1+\kappa}} U_{\Lb}^2(\gamma_{+}(0;\ub,u))\cdot 2\varrho^{\frac{2\sqrt{\kappa}}{1+\kappa}} U_{L'}^2(\gamma_{-}(0;\ub,u))\right)^{\frac{1+\kappa}{4\sqrt{\kappa}}}\\
\lesssim&\left( (|u(\gamma_{+}(0;\ub,u))|^{-2}a)^{\frac{2\sqrt{\kappa}}{1+\kappa}}b^2\cdot (|u(\gamma_{-}(0;\ub,u))|^{-2}a)^{\frac{2\sqrt{\kappa}}{1+\kappa}}b^2\right)^{\frac{1+\kappa}{4\sqrt{\kappa}}}\\
\leqslant& |u|^{-2}a\cdot b^\frac{1+\kappa}{\sqrt{\kappa}}\end{align*}
where we have used the initial bound of $\varrho,U_{\Lb}, U_{L'}$ and \eqref{ucompare2}. The same procedure gives the same bound for $\varrho^{-1}$:
$$\varrho^{-1}=\left(\varrho^{\frac{2\sqrt{\kappa}}{1+\kappa}} U_{\Lb}^2\cdot \varrho^{\frac{2\sqrt{\kappa}}{1+\kappa}} U_{L'}^2\right)^{-\frac{1+\kappa}{4\sqrt{\kappa}}}\lesssim |u|^{2}a\cdot b^\frac{1+\kappa}{\sqrt{\kappa}}$$
which improves \eqref{varrhoinverseb}. This implies that the bound \eqref{varrho} is true without assuming \eqref{varrhoinverseb}.
Similarly, we have
$$|U_{\Lb}|=\left(\varrho^{\frac{2\sqrt{\kappa}}{1+\kappa}} U_{\Lb}^2\Big/ \varrho^{\frac{2\sqrt{\kappa}}{1+\kappa}} U_{L'}^2\right)^{\frac{1}{4}}\lesssim a^{\frac{\sqrt{\kappa}}{1+\kappa}}\cdot b$$
and the same for $|U_{L'}|$. So we  have proved \eqref{U}. In particular, we can drop the factor $C^{\varepsilon}$
in \eqref{taub},  by using the bound \eqref{U} and go through the estimates there again. That is, it holds that
 \begin{align}\label{tau}\tau_\pm(\ub,u)\lesssim  \Omega_0^2(u)\ub a^\frac{\sqrt{\kappa}}{1+\kappa}b,\end{align}
which is an estimate of the length of $\gamma_{\pm}$ inside the spacetime region we consider.

Finally we consider the connection coefficients. From the equation \eqref{Dh2}, and the bounds \eqref{geometry}, \eqref{Lphi}, \eqref{varrho}, \eqref{U}, we have
\begin{align*}|h-h_0|\lesssim& \int_0^\delta\left|-r\Omega^{-2}(L\phi)^2-r\Omega^{-2}(1+\kappa)\varrho(U_L)^2\right|\D\ub\\\lesssim&\delta\Omega_0^{-2}|u|^{-1}a^2+\delta\Omega_0^2|u|^{-1}a\cdot b^\frac{1+\kappa}{\sqrt{\kappa}}(a^{\frac{\sqrt{\kappa}}{1+\kappa}}\cdot b)^2\\
\lesssim&\delta\Omega_0^{-2}|u|^{-1}a^2b^{\frac{(1+\sqrt{\kappa})^2}{\sqrt{\kappa}}}. \end{align*}
This proves \eqref{h-h0}. From the equation \eqref{DbhDhb2}, and the bounds \eqref{geometry}, \eqref{h}, \eqref{hb} and \eqref{varrho}, we have
$$|\hb+1|\lesssim\int_0^\delta\left|-\frac{\Omega^2(1+h\hb-r^2(1-\kappa)\varrho)}{r}\right|\D\ub\lesssim\delta|u|^{-1}ab^{\frac{1+\kappa}{\sqrt{\kappa}}}.$$
This proves \eqref{hb+1}. By \eqref{geometry}, \eqref{h}, \eqref{hb}, \eqref{Lphi}, \eqref{Lbphi} and \eqref{varrho}, the right hand side of equation \eqref{Dbomega2} is estimated by
\begin{align*}
&\left|\frac{\Omega^2(1+h\hb)}{r^2}-L\phi\Lb\phi-(1+\kappa)\varrho)\right|\\
\lesssim &|u|^{-2}+\delta|u|^{-3}a^2+|u|^{-2}ab^{\frac{1+\kappa}{\sqrt{\kappa}}}\lesssim |u|^{-2}ab^{\frac{1+\kappa}{\sqrt{\kappa}}},
\end{align*}
where in the last inequality we use \eqref{smallness1}. Then by the initial bound of $r\omega$ on $C_{u_0}$, the estimate \eqref{omega} of $\omega$  follows by integrating above bound relative to $u$. For $\omegab$, we will need its initial bound on $\Cb_0$:
$$|\omegab|_{\Cb_0}|=\left|\frac{\partial}{\partial u}\log\Omega_0\right|=\left|\frac{1}{2}|u|(1+\kappa)\varrho(U_{\Lb})^2\right|\lesssim |u|^{-1}ab^2,$$
where we use the equation \eqref{DbhbCb0} on $\Cb_0$. Then  by \eqref{Dbomega2}, we have
\begin{align*}|\omegab|\lesssim& |\omegab|_{\Cb_0}|+\int_0^{\delta}\left|\frac{\Omega^2(1+h\hb)}{r^2}-L\phi\Lb\phi-\frac{1}{2}(1+\kappa)\varrho\right|\D\ub\\
\lesssim&|u|^{-1}ab^2+\delta|u|^{-2}ab^{\frac{1+\kappa}{\sqrt{\kappa}}}\lesssim |u|^{-1}ab^2
\end{align*}
where we use \eqref{smallness1}. This proves \eqref{omegab}.
\end{proof}

We have closed the estimates of the geometric quantities $h,\hb,\omega,\omegab$, derivatives of the scalar function $L\phi,\Lb\phi$ and the fluid variables $\varrho, U$ in the region $(\ub,u)\in[0,\delta]\times[u_0,u_1]$ for $\delta, u_1$ satisfying \eqref{smallness1}. To construct a regular solution, we also need to consider the first derivatives of the fluid variables.

\begin{theorem}\label{apriori2}
There exists a constant $C>0$ depending on $\kappa$ (which may be larger than that in previous theorem) such that the following statement is true. Let $a,b\geqslant 1$ be such that there hold the same assumptions of Theorem \ref{apriori1}. Let $A\geqslant 1$ be such that
$$ \sup_{\Cb_0\cup C_{u_0}}\{ |u||\Lb\log\varrho|, |u||\Lb\log (-U_{\Lb'})|\}\leqslant A.$$
Then if in addition to \eqref{smallness1}, it holds
\begin{align}\label{smallness2}C\delta|u_1|^{-1}a^{\frac{\sqrt{\kappa}}{1+\kappa}}b^2A\leqslant 1,\end{align}
we will have
\begin{align}\label{Dbfluid}|u||\Lb\log\varrho|, |u||\Lb\log (-U_{\Lb'})|=|u| |\Lb\log (-U_{L})| \lesssim A,\end{align}
for $0\leqslant\ub\leqslant\delta, u_0\leqslant u\leqslant u_1<0.$
\end{theorem}
\begin{remark}
Here $\Lb'=\Omega^{-2}\Lb$ is geodesic. Recalling $L=\Omega^2 L'$, from the relation $U_{L}U_{\Lb'}=U_{L'}U_{\Lb}=1$, in fact we have $\log (-U_{\Lb'})=-\log (-U_{L})$. We estimate $\Db$ derivative of $\log(-U_{\Lb'})$ instead of $\log(-U_{\Lb})$ simply because estimating the latter one will involve an estimate of $\Lb\omegab$, which we want to avoid.
\end{remark}
\begin{remark}
By \eqref{acousticalnull}, the Euler equations \eqref{euler12}, \eqref{euler22}, and the estimate \eqref{Dbfluid}, we know that $D$ derivatives of $\log\varrho, \log(-U_{\Lb'}), \log(-U_{L})$ are also bounded. Then this implies that there are no shock singularities in the spacetime region we consider.
\end{remark}
\begin{proof}
We will commute $\Lb$ with the Euler equations.  To avoid estimating $\Lb\omegab$, we rewrite the Euler equation as
	\begin{align}\label{euler13}U_+\log\left(\varrho^{\frac{2\sqrt{\kappa}}{1+\kappa}} U_{\Lb'}^2\right)=2\Omega^{-2}\omega U_{\Lb}\sqrt{\frac{1+\sqrt{\kappa}}{1-\sqrt{\kappa}}}-\frac{2}{r}\sqrt{\frac{\kappa}{1-\kappa}}(U_{L'}\hb+U_{\Lb}h)\triangleq Q_+,\end{align}
	\begin{align}\label{euler23}U_-\log\left(\varrho^{\frac{2\sqrt{\kappa}}{1+\kappa}} U_{L}^2\right)=-2\Omega^{-2}\omega U_{\Lb}\sqrt{\frac{1-\sqrt{\kappa}}{1+\sqrt{\kappa}}}+\frac{2}{r}\sqrt{\frac{\kappa}{1-\kappa}}(U_{L'}\hb+U_{\Lb}h)\triangleq Q_-.\end{align}
 The left hand sides are about $U_{\Lb'}$ and $U_{L}$ instead of $U_{\Lb}$ and $U_{L'}$, and $\omega$ instead of $\omegab$ appears at the right hand sides. A simple way to derive these equations is taking the ``conjugates'' of \eqref{euler12} and \eqref{euler22}, that is, changing ``$+$'' to ``$-$'' and vice versa, ``$\Lb$'' to ``$L$'' and ``$L'$'' to ``$\Lb'$''.  One can also use directly the expressions of $U_{\pm}$ in \eqref{acousticalnull}, the original Euler equations \eqref{euler12}, \eqref{euler22} and write $\varrho^{\frac{2\sqrt{\kappa}}{1+\kappa}} U_{\Lb'}^2=\Omega^{-4}\varrho^{\frac{2\sqrt{\kappa}}{1+\kappa}} U_{\Lb}^2$ and $\varrho^{\frac{2\sqrt{\kappa}}{1+\kappa}} U_{L}^2=\Omega^4\varrho^{\frac{2\sqrt{\kappa}}{1+\kappa}} U_{L'}^2$.

Commuting $\Lb$, we have
	\begin{align}\label{Dbeuler1}U_+\left(\Lb\log\left(\varrho^{\frac{2\sqrt{\kappa}}{1+\kappa}} U_{\Lb'}^2\right)\right)=\Lb Q_++[U_+,\Lb]\left(\log\left(\varrho^{\frac{2\sqrt{\kappa}}{1+\kappa}} U_{\Lb'}^2\right)\right),\end{align}
	\begin{align}\label{Dbeuler2}U_-\left(\Lb\log\left(\varrho^{\frac{2\sqrt{\kappa}}{1+\kappa}} U_{L}^2\right)\right)=\Lb Q_-+[U_-,\Lb]\left(\log\left(\varrho^{\frac{2\sqrt{\kappa}}{1+\kappa}} U_{L}^2\right)\right).\end{align}
By direct computation, recalling \eqref{acousticalnull} and the relation $[L,\Lb]=0$, we have
\begin{align}\label{commutator}
[U_{\pm},\Lb]=\frac{1}{2}\sqrt{\frac{1\mp\sqrt{\kappa}}{1\pm\sqrt{\kappa}}} \Lb(U_{L'})\Lb+\frac{1}{2}\sqrt{\frac{1\pm\sqrt{\kappa}}{1\mp\sqrt{\kappa}}} \Lb(U_{\Lb'})L.
\end{align}
Using \eqref{acousticalnull} again and $U_{\Lb}U_{L'}=U_{\Lb'}U_{L}=1$, we have
\begin{align*}\frac{1}{2}\sqrt{\frac{1\pm\sqrt{\kappa}}{1\mp\sqrt{\kappa}}} \Lb(U_{\Lb'})L=&-\Lb(U_{\Lb'})\cdot\Omega^2\left(U_{L'}U_{\pm}+\frac{1}{2}\sqrt{\frac{1\mp\sqrt{\kappa}}{1\pm\sqrt{\kappa}}}(U_{L'})^2\Lb\right)\\
=&-\Lb\log(-U_{\Lb'})\left(U_{\pm}+\frac{1}{2}\sqrt{\frac{1\mp\sqrt{\kappa}}{1\pm\sqrt{\kappa}}}U_{L'}\Lb\right)\end{align*}
Plugging in the equations \eqref{euler13}, \eqref{euler23}, the equations \eqref{Dbeuler1}, \eqref{Dbeuler2} become
	\begin{align}\label{Dbeuler12}U_+\left(\Lb\log\left(\varrho^{\frac{2\sqrt{\kappa}}{1+\kappa}} U_{\Lb'}^2\right)\right)=P_+\Lb\log\left(\varrho^{\frac{2\sqrt{\kappa}}{1+\kappa}} U_{\Lb'}^2\right)+\Lb Q_+-(\Lb\log (-U_{\Lb'}) )Q_+,\end{align}
	\begin{align}\label{Dbeuler22}U_-\left(\Lb\log\left(\varrho^{\frac{2\sqrt{\kappa}}{1+\kappa}} U_{L}^2\right)\right)=P_-\Lb\log\left(\varrho^{\frac{2\sqrt{\kappa}}{1+\kappa}} U_{L}^2\right)+\Lb Q_--(\Lb\log(-U_{\Lb'}) )Q_-,\end{align}
where $Q_{\pm}$ are defined in \eqref{euler13}, \eqref{euler23} and
\begin{align*}P_\pm=&\frac{1}{2}\sqrt{\frac{1\mp\sqrt{\kappa}}{1\pm\sqrt{\kappa}}}\left(\Lb(U_{L'})-(\Lb\log(-U_{\Lb'}))U_{L'}\right)\\
=&-\sqrt{\frac{1\mp\sqrt{\kappa}}{1\pm\sqrt{\kappa}}}\left(\Lb\log(-U_{\Lb'})+\omegab)\right)U_{L'}.
\end{align*}

We only need to prove \eqref{Dbfluid} under the bootstrap assumptions:
\begin{align}\label{Dbfluidb}|u||\Lb\log\rho|, |u||\Lb\log (-U_{\Lb'})|\lesssim C^\varepsilon A.\end{align}
Using the estimates derived in Theorem \ref{apriori1}, and \eqref{h}, \eqref{hb}, we know
\begin{align*}|Q_{\pm}|\lesssim& |u|^{-1}\Omega_0^{-2}a^{\frac{\sqrt{\kappa}}{1+\kappa}+1}b^{\frac{1+\kappa}{\sqrt{\kappa}}+1}+|u|^{-1}\Omega_0^{-2}a^{\frac{\sqrt{\kappa}}{1+\kappa}}b\\\lesssim& |u|^{-1}\Omega_0^{-2}a^{\frac{\sqrt{\kappa}}{1+\kappa}+1}b^{\frac{1+\kappa}{\sqrt{\kappa}}+1}.\end{align*}
And then by \eqref{omegab}, \eqref{U} and \eqref{Dbfluidb}, we have
\begin{align}\label{Qpm}|(\Lb\log (-U_{\Lb'}) )Q_\pm|\lesssim C^{\varepsilon}|u|^{-2}\Omega_0^{-2}a^{\frac{\sqrt{\kappa}}{1+\kappa}+1}b^{\frac{1+\kappa}{\sqrt{\kappa}}+1}A,\end{align}
and 
\begin{align}\label{Ppm}|P_{\pm}|\lesssim C^{\varepsilon}|u|^{-1}a^{\frac{\sqrt{\kappa}}{1+\kappa}}b A+|u|^{-1}a^{\frac{\sqrt{\kappa}}{1+\kappa}}b\cdot ab^2.\end{align}
By equations  \eqref{DbhDhb2} and \eqref{Dbomega2}, and using \eqref{smallness1}, \eqref{h}, \eqref{hb} and the estimates derived in Theorem \ref{apriori1}, we have
$$|\Lb (\Omega^2h)|\lesssim |u|^{-1}a b^{\frac{1+\kappa}{\sqrt{\kappa}}},$$
$$|\Lb\omega|\lesssim |u|^{-2}a b^{\frac{1+\kappa}{\sqrt{\kappa}}},$$
furthermore, by \eqref{Dbhb2} and the relation $U_{\Lb}U_{L'}=1$, we have
$$|U_{L}\Lb(\Omega^{-2}\hb)|\lesssim |u|^{-1}a^{\frac{\sqrt{\kappa}}{1+\kappa}+2}b^{\frac{1+\kappa}{\sqrt{\kappa}}+1},$$
and by \eqref{Dbfluidb}, we have
$$|\Lb(U_{\Lb'})|=|\Omega^{-2}U_{\Lb}||\Lb\log (-U_{\Lb'})|\lesssim C^\varepsilon |u|^{-1}\Omega_0^{-2} a^{\frac{\sqrt{\kappa}}{1+\kappa}}bA,$$
$$|\Lb(U_{L})|=|\Omega^2U_{L'}||\Lb\log (-U_{L})|\lesssim C^\varepsilon|u|^{-1}\Omega_0^2a^{\frac{\sqrt{\kappa}}{1+\kappa}}bA.$$
By combining these estimates, we have
\begin{align}\label{LbQpm}|\Lb Q_{\pm}|\lesssim C^\varepsilon |u|^{-2}\Omega_0^{-2}a^{\frac{\sqrt{\kappa}}{1+\kappa}+1} b^{\frac{1+\kappa}{\sqrt{\kappa}}+1}A+ |u|^{-2}\Omega_0^{-2}a^{\frac{\sqrt{\kappa}}{1+\kappa}+2} b^{\frac{1+\kappa}{\sqrt{\kappa}}+1}.\end{align}

Using \eqref{tau}, the estimate \eqref{Ppm} integrating along $\gamma_{\pm}$ gives
$$\int_{0}^{\tau_{\pm}(\ub,u)}|P_{\pm}|\D\tau'\lesssim C^{\varepsilon}\ub |u|^{-1}a^{\frac{2\sqrt{\kappa}}{1+\kappa}}b^2 A+\ub|u|^{-1}a^{\frac{2\sqrt{\kappa}}{1+\kappa}+1}b^4\lesssim C^{\varepsilon-1}.$$
where we have used \eqref{smallness2} and \eqref{smallness1} in the last inequality. By choosing $C$ sufficiently large, we will have
\begin{align}\label{intPpm} \int_{0}^{\tau_{\pm}(\ub,u)}|P_{\pm}|\D\tau'\leqslant \log2.\end{align}
By integrating the estimates \eqref{Qpm} and \eqref{LbQpm}  along $\gamma_{\pm}$, we have
\begin{align}\nonumber&\int_{0}^{\tau_{\pm}(\ub,u)}|\Lb Q_\pm-(\Lb\log(-U_{\Lb'}) )Q_\pm|\D\tau'\\
\nonumber\lesssim& C^{\varepsilon} \ub|u|^{-2}a^{\frac{2\sqrt{\kappa}}{1+\kappa}+1}b^{\frac{(1+\sqrt{\kappa})^2}{\sqrt{\kappa}}} A+\ub|u|^{-2}a^{\frac{2\sqrt{\kappa}}{1+\kappa}+1}b^{\frac{(1+\sqrt{\kappa})^2}{\sqrt{\kappa}}}\\ 
\label{intLbQpm}\lesssim &C^{\varepsilon-1}|u|^{-1}A\end{align}
where we have used \eqref{smallness1} in the last inequality. \eqref{smallness2} is not needed here because the inhomogeneous term $\Lb Q_--(\Lb\log(-U_{\Lb'}) )Q_-$ is linear to the first order derivative.

Using \eqref{ucompare}, \eqref{intPpm} and \eqref{intLbQpm}, by Gronwall's inequality along $\gamma_\pm$, we have
\begin{align*}&\left|\Lb\log\left(\varrho^{\frac{2\sqrt{\kappa}}{1+\kappa}} U_{\Lb'}^2\right)\right|\\
\leqslant&2\left|\Lb\log\left(\varrho^{\frac{2\sqrt{\kappa}}{1+\kappa}} U_{\Lb'}^2\right)\right|_{\Cb_0\cup C_{u_0}}+2\int_{0}^{\tau_{+}(\ub,u)}|\Lb Q_+-(\Lb\log(-U_{\Lb'}) )Q_+|\D\tau'\\
\lesssim& |u|^{-1}A\end{align*}
and similarly
$$\left|\Lb\log\left(\varrho^{\frac{2\sqrt{\kappa}}{1+\kappa}} U_{L}^2\right)\right|\lesssim |u|^{-1}A.$$
So
$$|u||\Lb\log\rho|\lesssim|u| \left|\Lb\log\left(\varrho^{\frac{2\sqrt{\kappa}}{1+\kappa}} U_{\Lb'}^2\right)+\Lb\log\left(\varrho^{\frac{2\sqrt{\kappa}}{1+\kappa}} U_{L}^2\right)\right|\lesssim A,$$
and
$$|u||\Lb\log U_{\Lb'}|\lesssim |u|\left|\Lb\log\left(\varrho^{\frac{2\sqrt{\kappa}}{1+\kappa}} U_{\Lb'}^2\right)-\Lb\log\left(\varrho^{\frac{2\sqrt{\kappa}}{1+\kappa}} U_{L}^2\right)\right|\lesssim A,$$
which completes the proof.
\end{proof}

The bounds for higher order derivatives of $h,\hb,\omega,\omegab, L\phi, \Lb\phi$ and $\varrho, U$ can be estimated by commuting more $L$ and $\Lb$ to the equations. The resulting equations are linear relative to the higher order derivatives, so the estimates can be derived directly in terms of the corresponding bounds on $\Cb_0\cup C_{u_0}$, in the region $(\ub,u)\in[0,\delta]\times[u_0,u_1]$ for $\delta,u_1$ satisfies \eqref{smallness1} and \eqref{smallness2}. The higher order bounds need not to have a self-similar form as those in Theorem \ref{apriori1} and \ref{apriori2}.  Then a regular solution of the system \eqref{EinsteinSE} in double null coordinates in this region can be constructed by standard argument (see for example \cite{G-H-J23} without scalar field).

\section{Formation of trapped surfaces and instability}

In the last section we will find the trapped surface formed in the solution by requiring that  the initial scalar field on $C_{u_0}$ satisfies certain lower bound condition.

\begin{theorem}\label{fots}There is a universal constant $c_1>0$ such that the following is true. Let $f(u)$ be a positive function defined on $u\in [u_0,0)$. Suppose that there is a $u_1\in (u_0,0)$ such that,
\begin{align}\label{Omega0<=f}\Omega_0(u_1)\leqslant f(u_1),\end{align}
and it holds
\begin{align}\label{smallness3}C(f(u_1))^{\frac{2}{3}}a^2b^{\frac{(1+\sqrt{\kappa})^2}{\sqrt{\kappa}}}A\leqslant 1,\end{align}
\begin{align}\label{lowerbound}\int_0^{\delta}|rL\phi\big|_{C_{u_0}}|^2\D\ub\geqslant c_1f^2(u_1)|u_1|,\end{align}
where $C$ is determined in Theorem \ref{apriori2} and 
\begin{align}\label{deltau1}\delta a^{\frac{2}{3}}=(f(u_1))^{\frac{2}{3}}|u_1|.\end{align}
Then the solution of the system \eqref{EinsteinSE} exists in $(\ub,u)\in[0,\delta]\times[u_0,u_1]$ and $S_{\delta,u_1}$ is a closed trapped surface.
\end{theorem}
\begin{proof}
By \eqref{smallness3} and \eqref{deltau1}, we have 
$$C\delta|u_1|^{-1}a^{2}b^{\frac{(1+\sqrt{\kappa})^2}{\sqrt{\kappa}}}A\leqslant 1,$$
then \eqref{smallness1} and \eqref{smallness2} hold, and hence the solution exists and all estimates of  Theorems \ref{apriori1}, \ref{apriori2} hold.

From the wave equation, 
$$\Db(rL\phi)=-\Omega^2 h\Lb\phi,$$
we have, for some universal constant $c,c'>0$, using \eqref{Lbphi} and \eqref{h},
$$|rL\phi\big|_{C_{u_1}}-rL\phi\big|_{C_{u_0}}|\leqslant c'\int_{u_0}^{u_1}|-\Omega^2 h\Lb\phi|\D u\leqslant c\delta|u_1|^{-1}a.$$
Thus we have
$$|rL\phi\big|_{C_{u_1}}|\geqslant | rL\phi\big|_{C_{u_0}}|-c\delta|u_1|^{-1}a,$$
and then
$$|rL\phi\big|_{C_{u_1}}|^2\geqslant \frac{1}{2}| rL\phi\big|_{C_{u_0}}|^2-(c\delta|u_1|^{-1}a)^2.$$
By integrating relative to $\ub\in[0,\delta]$,  we have
\begin{align*}&\int_0^\delta|rL\phi\big|_{C_{u_1}}|^2\D\ub\geqslant \frac{1}{2}\int_0^{\delta}|rL\phi\big|_{C_{u_0}}|^2\D\ub-c^2\delta^3|u|^{-2}a^2\\\geqslant &\frac{c_1}{2}f^2(u_1)|u_1|-c^2f^2(u_1)|u_1|\geqslant\frac{c_1}{4}f^2(u_1)|u_1|,\end{align*}
where we have used \eqref{deltau1} at the second step and  set $c_1=\min\{4c^2,2^6\}$.

From the equation \eqref{Dh2},
$$Dh=-r\Omega^{-2}(L\phi)^2-r\Omega^{-2}(1+\kappa)\varrho(U_L)^2\leqslant-r\Omega^{-2}(L\phi)^2 ,$$
we have, on $C_{u_1}$, using \eqref{geometry} and \eqref{Omega0<=f},
\begin{align*}h-h_0\leqslant& -2^{-3}\Omega_0^{-2}(u_1)|u_1|^{-1}\int_0^\delta|rL\phi\big|_{C_{u_1}}|^2\D\ub\\
\leqslant&-2^{-3}f^{-2}(u_1)|u_1|^{-1}\int_0^\delta|rL\phi\big|_{C_{u_1}}|^2\D\ub\\\leqslant& -c_1\cdot2^{-5}\leqslant-2,\end{align*}
since $h_0\leqslant 1$, we have $h<0$, which completes the proof.

\end{proof}

We finally turn to the proof of Theorem \ref{main}. 
\begin{proof}[Proof of Theorem \ref{main}]
It is easy to see there are $a,b,A\geqslant1$ depending on $\kappa$ and $B$ in Theorem \ref{main} such that the assumptions of bounds for initial data on $\Cb_0$ in Theorems \ref{apriori1}, \ref{apriori2} hold, and the bounds for initial data on $C_{u_0}$ are trivially true, by choosing possibly larger $a,b$ and $A$. Therefore we can apply Theorem \ref{fots}. We will construct $\phi_t$ verifying all  requirements in Theorem \ref{fots}.

From the assumptions of Theorem \ref{main}, we can find $\beta>0$ depending on $\kappa$ and $B$ such that the right hand side of \eqref{DbhbCb0}, the restriction of \eqref{Dbhb2} on $\Cb_0$, is bounded from above by $-\frac{\beta}{|u|}$.  We therefore have, for $u\in [u_0,0)$,
$$\beta(\log|u_0|-\log|u|)\leqslant \log\Omega_0(u_0)-\log\Omega_0(u).$$
So we will have, for some $c_\beta$ depending on $\Omega_0(u_0)$ and $\beta$, 
$$\Omega_0(u)\leqslant c_\beta|u|^{\beta}, u\in [u_0,0).$$

Set $f(u)=c_\beta|u|^\beta$. Note that under the relation \eqref{deltau1}, the condition \eqref{lowerbound} becomes
$$\int_0^{\delta}|rL\phi\big|_{C_{u_0}}|^2\D\ub\geqslant C(c_1,c_\beta,a) \delta\cdot\delta^{\frac{4\beta}{2\beta+3}},$$
for some constant $C(c_1,c_\beta,a)$. Therefore, given any $\alpha\in\left(0,\frac{2\beta}{2\beta+3}\right)$, and any $t$ sufficiently small, we can find $\delta_t>0$ depending on various constants, including $t$, such that
\begin{align}\label{deltat}\int_0^{\delta_t} (t\ub^{\alpha})^2\D\ub\geqslant 2C(c_1,c_\beta,a) \delta_t\cdot\delta_t^{\frac{4\beta}{2\beta+3}},\end{align}
and $u_{1t}$ defined by $\delta_t a^{\frac{2}{3}}=(f(u_{1t}))^{\frac{2}{3}}|u_{1t}|$ as in \eqref{deltau1} verifies \eqref{smallness3}.
We moreover choose a $\tilde{\delta}_t\in(0,\delta_t)$ such that
\begin{align*}\int_{\tilde{\delta}_t}^{\delta_t} (t\ub^{\alpha})^2\D\ub = \frac{1}{2}\int_0^{\delta_t}(t\ub^{\alpha})^2\D\ub\geqslant C(c_1,c_\beta,a) \delta_t\cdot\delta_t^{\frac{4\beta}{2\beta+3}}.\end{align*}
Now we may choose a smooth $\phi_t$ such that 
$$rL\phi_t=\begin{cases}0, &0\leqslant\ub\leqslant\frac{\tilde{\delta}_t}{2}\\\text{monotonically increasing and derivative $\leqslant 2\alpha t\ub^{\alpha-1}$},  &\frac{\tilde{\delta}_t}{2}\leqslant \ub\leqslant\tilde{\delta}_t\\  t\ub^\alpha,& \tilde{\delta}_t\leqslant\ub\leqslant 1\end{cases}.$$
Note that for all small $t$, $a$ can be chosen independent of $t$. Then such $\phi_t$ verifies all requirements in   Theorem \ref{fots} and hence $S_{\delta_t,u_{1t}}$ is a closed trapped surface in the corresponding maximal development. 

To show the last statement in Theorem \ref{main}, given any such $\phi_t$ constructed above, let $g$ be a smooth function in $\ub\in[0,1]$ be such that $g(0)=0$ and 
$$\|g-rL\phi_t\|_{C^\alpha[0,1]}<t^2.$$
Then in particular we have
$$|g(\ub)-rL\phi_t(\ub,u_0)|<t^2\ub^\alpha, \ub\in[0,1],$$
we then have
\begin{align*}\int_0^{\delta_t}|g(\ub)|^2\D \ub\geqslant& \frac{1}{2}\int_0^{\delta_t}|rL\phi_t(\ub,u_0)|^2\D\ub-\int_0^{\delta_t}t^4\ub^{2\alpha}\D\ub\\\geqslant& \frac{1}{2}\int_{\tilde{\delta}_t}^{\delta_t}t^2\ub^{2\alpha}\D\ub-\int_0^{\delta_t}t^4\ub^{2\alpha}\D\ub\\
=&\frac{1}{4}\int_0^{\delta_t}t^2\ub^{2\alpha}\D\ub-\int_0^{\delta_t}t^4\ub^{2\alpha}\D\ub\\
\geqslant&\frac{1}{8}\int_0^{\delta_t}t^2\ub^{2\alpha}\D\ub\end{align*}
for $t$ sufficiently small. By choosing $\delta_t$ smaller such that the number $2$ is replaced by $8$ in \eqref{deltat}, we then have
$$\int_0^{\delta_t}|g(\ub)|^2\D\ub\geqslant C(c_1,c_\beta,a) \delta\cdot\delta^{\frac{4\beta}{2\beta+3}}.$$
Then if we set another initial data $rL\phi\big|_{C_{u_0}}=g$, it still verifies all requirement in Theorem \ref{fots} and hence a close trapped surface still forms. Then the open set in the last statement in Theorem \ref{main} can be chosen as
$$\bigcup_{\text{all sufficiently small}\ t}\{\text{smooth initial data}\ rL\phi| \|rL\phi-rL\phi_t\|_{C^\alpha[0,1]}<t^2\}$$
which is open as a union of open sets. The proof is then completed.
\end{proof}

\end{document}